\documentclass[11pt]{article}%

\IfFileExists{sariel_computer.sty}{\def\sarielComp{1}}{}
\ifx\sarielComp\undefined%
\newcommand{\SarielComp}[1]{}
\newcommand{\NotSarielComp}[1]{#1}%
\else
\newcommand{\SarielComp}[1]{#1}%
\newcommand{\NotSarielComp}[1]{}%
\fi
\newcommand{\IfPrinterVer}[2]{#2}%

\usepackage[cm]{fullpage}%
\usepackage{amsmath}%
\usepackage{amssymb}%
\usepackage{xcolor}%

\usepackage[amsmath,thmmarks]{ntheorem}%
\theoremseparator{.}%

\usepackage{graphicx}
\usepackage{titlesec}%
\titlelabel{\thetitle. }%
\usepackage{xcolor}%
\usepackage{mleftright}%
\usepackage{xspace}%
\usepackage{paralist}%
\usepackage{hyperref}%

\newcommand{\hrefb}[3][black]{\href{#2}{\color{#1}{#3}}}%

\IfPrinterVer{%
   \usepackage{hyperref}%
}{%
   \usepackage{hyperref}%
   \hypersetup{%
      breaklinks,%
      ocgcolorlinks, colorlinks=true,%
      urlcolor=[rgb]{0.25,0.0,0.0},%
      linkcolor=[rgb]{0.5,0.0,0.0},%
      citecolor=[rgb]{0,0.2,0.445},%
      filecolor=[rgb]{0,0,0.4},
      anchorcolor=[rgb]={0.0,0.1,0.2}%
   }
}

\theoremseparator{.}%

\theoremstyle{plain}%
\newtheorem{theorem}{Theorem}[section]

\theoremstyle{plain}%
\theoremheaderfont{\sf} \theorembodyfont{\upshape}%
\newtheorem*{remark:unnumbered}[theorem]{Remark}%
\newcommand{\myqedsymbol}{\rule{2mm}{2mm}}

\theoremheaderfont{\em}%
\theorembodyfont{\upshape}%
\theoremstyle{nonumberplain}%
\theoremseparator{}%
\theoremsymbol{\myqedsymbol}%
\newtheorem{proof}{Proof:}%

\newcommand{\atgen}{\symbol{'100}}
\newcommand{\SarielThanks}[1]{\thanks{Department of Computer Science;
      University of Illinois; 201 N. Goodwin Avenue; Urbana, IL,
      61801, USA; {\tt sariel\atgen{}illinois.edu}; {\tt
         \url{http://sarielhp.org/}.} #1}}

\newcommand{\HLink}[2]{\hyperref[#2]{#1~\ref*{#2}}}
\newcommand{\HLinkSuffix}[3]{\hyperref[#2]{#1\ref*{#2}{#3}}}

\newcommand{\figlab}[1]{\label{fig:#1}}
\newcommand{\figref}[1]{\HLink{Figure}{fig:#1}}

\newcommand{\thmlab}[1]{{\label{theo:#1}}}
\newcommand{\thmref}[1]{\HLink{Theorem}{theo:#1}}

\newcommand{\seclab}[1]{\label{sec:#1}}
\newcommand{\secref}[1]{\HLink{Section}{sec:#1}}

\providecommand{\eqlab}[1]{}%
\renewcommand{\eqlab}[1]{\label{equation:#1}}

\newcommand{\remove}[1]{}%

\newcommand{\pth}[2][\!]{\mleft({#2}\mright)}%

\newcommand{\floor}[1]{\left\lfloor {#1} \right\rfloor}

\newcommand{\cardin}[1]{\left| {#1} \right|}%

\renewcommand{\Re}{\mathbb{R}}%
\providecommand{\Mh}[1]{#1}%
\newcommand{\etal}{\textit{et~al.}\xspace}

\providecommand{\Mh}[1]{#1}%
\newcommand{\Gr}{\Mh{G}}%
\newcommand{\GrA}{\Mh{H}}%

\newcommand{\VV}{\Mh{V}}
\newcommand{\VX}[1]{\VV\pth{#1}}%

\newcommand{\EX}[1]{E\pth{#1}}%

\newcommand{\Tree}{\Mh{T}}%
\newcommand{\TreeA}{\Mh{S}}%
\newcommand{\RSet}{\Mh{M}}%
\newcommand{\rsz}{\Mh{\mu}}%

\newcommand{\Sep}{\Mh{Z}}%

\newcommand{\email}[1]{\href{mailto:#1}{#1}}%
\begin{document}

\title{Separators for Planar Graphs that are Almost Trees}

\author{%
   Linda Cai%
   \thanks{\email{tcai4@illinois.edu}.}%
   \and %
   Sariel Har-Peled\SarielThanks{Work on this paper was partially
      supported by a NSF AF award CCF-1421231.  }%
   \and %
   Simiao Ye%
   \thanks{\email{sye11@illinois.edu}.}%
}

\date{\today}

\maketitle

\begin{abstract}
    We prove that a connected planar graph with $n$ vertices and
    $n+\rsz$ edges has a vertex separator of size $O( \sqrt{\rsz} + 1)$,
    and this separator can be computed in linear time.
\end{abstract}


\section{Result}

We first provide a relatively self-contained proof of the claim. A
shorter proof using known tools is described in \secref{shorter}.

\begin{theorem}
    \thmlab{sep:tree}%
    Let $\Gr$ be a vertex connected planar graph with $n$ vertices and
    $n+\rsz$ edges, and weights $w:\VX{\Gr} \rightarrow \Re^+$ on the
    vertices.  Then, one can compute, in linear time, a vertex
    separator for $\Gr$ of size $O( \sqrt{\rsz} + 1)$
\end{theorem}

\begin{proof}
    The proof in depicted in \figref{proof:in:pics}.  Assume
    $\rsz > 0$, as otherwise the result is immediate.  Let $\Tree$ be
    a spanning tree of $\Gr$, and let $\RSet$ be the remaining
    $\rsz + 1$ edges; i.e., $\RSet = \EX{\Gr} \setminus
    \EX{\Tree}$. Let $\TreeA$ be the minimal subtree of $\Tree$ that
    contains all the vertices of $\VX{\RSet}$. Let $U$ be the set of
    all vertices of $\TreeA$ that are either of degree three (or
    higher), or belong to $\VX{\RSet}$. Since the only leafs of
    $\TreeA$ are vertices of $\VX{\RSet}$, it follows that
    $\cardin{U} \leq 2 \cardin{\VX{\RSet}} \leq 4 (\rsz + 1)$. %

    Decompose $\TreeA$ into maximal set of paths, such that their
    endpoints are in $U$ (and no vertex of $U$ is contained in the
    interior of such a path), and let $\Pi$ be this collection of
   paths. Observe that
    $\cardin{\Pi} \leq \cardin{U} = O(\rsz )$.

    Consider assigning the weight of every vertex of
    $\Tree \setminus \TreeA$ to its nearest vertex in $\TreeA$.  As
    such, under the new weights $w'$, we have $w'(\TreeA) = w(\Tree)$,
    where a weight of a path is the total weight of the vertices in
    its interior.
    
    Consider the planar graph $(U, \Pi \cup \RSet)$ with weights on
    the edges and vertices. It has $O(\rsz )$ vertices and $O(\rsz)$
    edges. As such, by Lipton and Tarjan planar separator theorem
    \cite{lt-stpg-79}, it has a balanced separator
    $\Sep \subseteq \VX{\Gr}$ of size
    $O( \sqrt{\cardin{U}}) = O( \sqrt{\rsz})$, as desired.

    A vertex $u \in \Sep$, with weight $w'(u)$ might see its weight
    decrease to $w(u)$ in the original graph, because of various trees
    attached to $u$ with total weight $w'(u) - w(u)$. Since $u$ is in
    the separator, all these trees get separated when $\Sep$ is
    removed. Namely, the separator set $\Sep$ is still a balanced
    separator in $\Gr$.
    
    The only case where the above argument fails, is if the computed
    separator has a single vertex $z$ with majority of the weight
    (i.e., $w'(z) \geq (2/3)w'(U)$. Furthermore, the vertex $z$ might
    have a tree $\Tree' $ attached to it with weight exceeding
    $(2/3)w(\Gr)$. But this can be fixed by just adding the vertex
    separator $z'$ of $\Tree'$ to the separator set, thus implying the
    claim (in particular, this case, the separator is made out of two
    vertices $z$ and $z'$).

    It is easy to verify that the above algorithm can be implemented
    in linear time.
\end{proof}

The above proof works for any family of graphs that is minor closed
and has a sublinear sized separator. In particular, if a graph in this
family with $\rsz$ vertices has separator of size $g(\rsz)$, then the
above proof implies that a connected graph in this family with $n$
vertices and $n+\rsz$ edges, has a separator of size $O(g(\rsz))$.


\newcommand{\Frame}[1]{%
   \noindent%
   \begin{minipage}{0.31\linewidth}
       \smallskip
       
       \centerline{\includegraphics[page=#1,width=0.9\linewidth]{figs/sep}}%
       
       
       \smallskip%
   \end{minipage}%
} \newcommand{\FText}[1]{%
   \begin{minipage}[c]{0.30\linewidth}
       #1 \smallskip
   \end{minipage}%
}

\begin{figure}
    \begin{tabular}{c|c|c}
      \Frame{1}
      &
        \Frame{2}
      &
        \Frame{3}
      \\%
      \FText{(A) A tree with $\rsz+1$ additional edges.}
      &
        \FText{(B) The endpoints of the additional edges (i.e., $\RSet$).}%
      &
        \FText{(C) The spanning tree $\TreeA$ of these endpoints. }
      \\
      \hline%
      \Frame{4}%
      &%
        \Frame{5}
      &
        \Frame{6}
      \\
      \FText{(D) The additional vertices of degree $3$ in $\TreeA$
      (i.e., $U$).}
      &
        \FText{(E) The planar graph $\GrA$ induced by all the vertices of
        interest by the tree $S$, and the $\rsz+1$ additional edges. This
        graph has $O(\rsz)$ vertices and edges.%
        \smallskip
        }
      &
        \FText{(F) The weighted vertex separator }%
      \\
      \hline%
      \Frame{7}{}
      &%
        \Frame{8}{}
      &
        \Frame{9}{}
      \\%
      \FText{(G) There might be  one collapsed tree that has
      the majority of the mass of the original graph.}
      &
        \FText{(H) Add the vertex separator of this tree to the separator set.}
      &
        \FText{(I) The resulting separator in the original graph.}
    \end{tabular}
    \caption{A proof in pictures of \thmref{sep:tree}.}
    \figlab{proof:in:pics}%
\end{figure}


\subsection{A proof using known tools}
\seclab{shorter}

We next provide a shorter proof using known tools -- the proof was
pointed out to us by Chandra Chekuri.

\begin{proof}
    Let $\Tree$ be a spanning tree of $\Gr$, and let
    $\RSet = \EX{\Gr} \setminus \EX{\Tree}$ be the remaining set of
    $\rsz+1$ edges.

    Let $g$ be the treewidth of $\Gr$. By Robertson \etal
    \cite[Theorem 6.2]{rst-qexpg-94}, the graph $\Gr$ has a grid minor
    of size $g' \times g'$, where $u = \floor{ (g-5)/6}$. But then,
    $\Gr$ must contain $\floor{u/2}^2$ vertex disjoint cycles. However
    $\Gr$ contains at most $\rsz+1$ disjoint cycles, since every cycle
    must contain an edge of $\RSet$. It follows that $g^2 =
    O(\rsz)$. Namely, $g = O(\sqrt{\rsz})$. A graph with tree width
    $g$, has a separator of size $O(g)$, thus implying the graph has a
    separator of size $O( \sqrt{\rsz} +1)$. ~
\end{proof}

The above argument is well known. see Demaine \etal \cite[Theorem
4.4]{dfht-spalbg-05},


\paragraph*{Acknowledgments.}

We thank Chandra Chekuri for pointing out the proof in
\secref{shorter}, and Daniel Lokshtanov for relevant references.

 


\end{document}